\documentclass[10pt,showpacs,letterpaper,aps,pra,onecolumn,notitlepage,superscriptaddress]{revtex4-1}

\usepackage{amssymb,amsmath,amsfonts,amsthm,bbm,graphicx}

\def\R{{\mathbb R}}

\def \i {{\rm i}}

\def \d {{\rm d}}

\def \om {{\omega }}

\def \beq { \begin{equation}}
\def \eeq { \end{equation}}

\def \( { \Bigl( }
\def \) {\Bigr) }
\def \CdUnmasked {{G_d^-}}
\def \CZeroUnmasked {{G^-}}

\def \IndicatorFunction {{\mathbbm{1}}}
\def\Gaussian#1#2{{N^{#1}(#2)}}
\def\GaussianNoParm#1{{N^{#1}}}

\def\sig{{\sigma}}

\newtheorem*{thmnonum}{Theorem}
\newtheorem*{mainthm*}{Main Theorem}

\newtheorem{lemma}{Lemma}[section]

\begin{document}

\title{Spatial coherence of fields from generalized sources in the Fresnel regime}



\author{Andre Beckus}
\affiliation{Department of Electrical and Computer Engineering, University of Central Florida, Orlando, FL 32816, USA}
\author{Alexandru Tamasan}
\affiliation{Department of Mathematics, University of Central Florida, Orlando, FL 32816, USA}
\author{Aristide Dogariu}
\affiliation{CREOL, The College of Optics \& Photonics, University of Central Florida, Orlando, FL 32816, USA}
\author{Ayman F. Abouraddy}
\affiliation{CREOL, The College of Optics \& Photonics, University of Central Florida, Orlando, FL 32816, USA}
\author{George K. Atia}
\affiliation{Department of Electrical and Computer Engineering, University of Central Florida, Orlando, FL 32816, USA}






\begin{abstract}
Analytic expressions of the spatial coherence of partially coherent fields propagating in the Fresnel regime in all but the simplest of scenarios are largely lacking and calculation of the Fresnel transform typically entails tedious numerical integration.
Here, we provide a closed-form approximation formula for the case of a generalized source obtained by modulating the field produced by a Gauss-Shell source model with a piecewise constant transmission function, which may be used to model the field's interaction with objects and apertures. The formula characterizes the coherence function in terms of the coherence of the Gauss-Schell beam propagated in free space and a multiplicative term capturing the interaction with the transmission function. 
This approximation holds in the regime where the intensity width of the beam is much larger than the coherence width under mild assumptions on the modulating transmission function. The formula derived for generalized sources lays the foundation for the study of the inverse problem of scene reconstruction from coherence measurements.  
\end{abstract}

\maketitle

\section{Introduction}

Spatial coherence of optical fields is assessed through cross-correlations between the random field-fluctuations at pairs of points in space. The coherence function $G$ relating two points $x_1,x_2$ in a quasi-monochromatic scalar field is given by $G(x_1,x_2)\!=\!\langle U(x_{1})U^{*}(x_{2})\rangle$, where $\langle\cdot\rangle$ represents an ensemble average, and $U(x)$ is one realization from the ensemble \cite{WOLF19813}. This coherence function can be obtained by various measurement strategies, e.g.,
through the use of double slits \cite{ThompsonPartiallyCoherent,DivittNonParallelSlits,DivittSunlight2015,CoherenceDMD,Kondakci17OE}, non-redundant arrays of apertures \cite{Mejia07OC,Gonzales11JOSAA}, lateral-shearing Sagnac and reversed-wavefront Young interferometers \cite{Iaconis96OL,Cheng00JMO,Santarsiero06OL}, microlens arrays \cite{Stoklasa14NC}, and phase-space methods \cite{Cho12OL,Wood14OL,Sharma16OE}.
The intensity of the field $I$ is subsumed in the coherence function and lies along its diagonal $I(x)\!=\!G(x,x)$. In this sense, the coherence function provides a complete description of a partially coherent field, whereas the intensity alone of course does not. Recently, experiments investigating the effects of one or two simple objects on the coherence functions of partially coherent sources confirmed that the position and size of an object are discernible from the coherence function alone \cite{Kondakci17OE,ElHalawany17}.
However, interpreting these measurements remains a challenge.
This motivates the work of this paper which seeks to develop effective mathematical tools for studying the process of coherence propagation.

Anywhere but in very special cases, the free evolution of coherence functions cannot be obtained analytically in closed form. Even if such a solution is found, once the field scatters off an object, further field evolution can only be evaluated numerically. For example, the generic Gauss-Schell model for a partially coherent field approximates the characteristics of the radiation produced by a wide range of optical sources. Furthermore, such a model admits a tractable analytical treatment of its free evolution \cite{FRIBERG1982383,GORI1983149}, or even for long-range propagation through turbid media as long as no size restrictions are involved \cite{SALEM2003261,Baleine:03}. However, once the intensity profile is modified by passage through a finite aperture (see Section 5.7 of \cite{Goodman:85}), transmittance through a partially transparent medium, or scattering off an object, the subsequent evolution of the coherence function no longer resembles the initial Gauss-Schell model. Instead, calculation of the propagated coherence function is accomplished using a double diffraction integral \cite{OpticalCoherenceBook}, which incurs a high computational cost, and -- crucially -- does not avail a suitable framework for the analysis of the inverse problem from coherence measurements. We call the field produced by such a secondary source, the original coherence function modulated with an arbitrary amplitude profile, a `generalized source'. Generalized sources such defined have bearing on various scenarios of practical importance. Most notably, the transmission function is well-suited to model a beam's interaction with objects or apertures.

There exist techniques that can help reduce the computational complexity, such as accelerating the calculation of the Fresnel integrals through the use of the Fast Fourier Transform \cite{Davis:12}, avoiding full computation of the Fresnel integrals \cite{Gbur2010285}, or exploiting the coherent communication modes of the propagation kernel itself in which the field is expanded \cite{Martinsson:07}. Another strategy involves carrying out a singular expansion of the source in terms of coherent modes to take advantage of the simpler coherent propagation integrals \cite{Gbur2010285,Shirai:03}, but the calculation of the modes is beam-specific \cite{WOLF19813,GORI1983149,Gori:08,Borghi:98,Vahimaa:06} and the number of required modes increases with reduced field coherence \cite{Vahimaa:06}. An altogether different numerical strategy makes use of ray-tracing \cite{Vicent2002101}, which can outperform Fresnel integration by limiting the number of rays \cite{Petruccelli:08}.

In this paper, we obtained a closed-form expression for the spatial coherence function of partially coherent fields propagating from generalized sources in the Fresnel regime, which reduces the computational complexity and affords a favorable ground for the study of inverse problems. We focus on a one-dimensional model in which the field is assumed to vary only along one transverse direction by a piecewise constant transmission function, but the concepts developed herein are naturally extendable to higher dimensions. Our closed-form solution characterizes the coherence from generalized sources in terms of a \textit{conjugated Hilbert transform} \cite{HilbertTransformsBook,EpsteinMedicalImagingBook}, a modified form of the Hilbert transform in which a function is first modulated by a linear phase, transformed, and then modulated by a conjugated phase. Some mild restrictions must be satisfied for this approach to succeed; e.g., the transverse coherence width must be at least one order of magnitude larger than the wavelength, but narrower than features of the transmission function of the generalized source.  A distinguishing feature of our approach is that the parameters of the source appear explicitly in the closed-form expression of the generalized source. For this reason, the results presented set the stage for the inverse problem in which reconstruction of a generalized source is intended from coherence measurements.

The paper is organized as follows. In Section \ref{section:prop_model}, we present the Fresnel propagation in a rotated coordinate system that serves as the basis of our work. Then, in Section \ref{section:prop_complex_source} we formally define generalized sources, provide details on the validity conditions that must be satisfied for our approach to succeed, and present the main theorem. Examples of generalized sources and numerical results obtained using the main theorem are illustrated in Section \ref{section:Examples}.  In the Discussion, we discuss the relevance of the main result to the inverse problem, and highlight possible extensions to the main theorem. Technical details such as the evaluation of the Fourier transform of a truncated Gaussian field and the proof of the main theorem are presented in Appendix \ref{section:math_tools} and Appendix \ref{section:MainResultAppendix}, respectively.

\section{Free Fresnel propagation in a rotated coordinate system}\label{section:prop_model}

Given a planar source located at $z\!=\!0$ with coherence function $G(x_1',x_2')$, after propagating a distance $d$ in the Fresnel regime, the coherence function becomes
\begin{align}\label{fresnel}
\iint G(x_1', x_2') \,
h(x_{1},x_{1}') \,
h^*(x_{2},x_{2}') \,
\;dx_1' \, dx_2',
\end{align}

where the Fresnel propagator $h$ is given by
\begin{align}\label{fresnelkernel}
h(x_{1},x_{1}')=\frac{\exp(ikd)}{\sqrt{i\lambda d}}\exp\left\{i\frac{k}{2d}(x_{1}-x_{1}')^2\right\};
\end{align}
here $\lambda$ is the wavelength and $k$ is the wavenumber.

Because the width of the intensity along the $x_{1}\!=\!x_{2}$ direction is usually significantly larger than the coherence width along the $x_{1}\!=\!-x_{2}$ direction, it is advantageous to work with rotated coordinates
\begin{align} \label{eqn:rotatedcoordinates}
y_1'=\frac{x_1'+x_2'}{2}, \qquad y_2'=\frac{x_1'-x_2'}{2}, 
\end{align}
and similarly for the unprimed coordinates. We refer to $y_1$  and $y_2$ hereon as the {\em intensity} and {\em coherence} coordinates, respectively. These rotated coordinates are illustrated in Fig.~\ref{Fig1:Overview}. Using these coordinates, the Fresnel model of coherence propagation becomes
\begin{align}\label{fresneliny}
G_d(y_1,y_2)=
\frac{1}{2\pi\ell^{2}}
\iint G(y_1', y_2') \, \mathcal{L}(y_1,y_1',y_2,y_2') \, dy_1' \, dy_2',
\end{align}
with $\ell\!=\!\sqrt{d/2k}$ and the kernel is
\begin{align}
\mathcal{L}(y_1,y_1',y_2,y_2') = \exp \left\{i(y_1-y_1')(y_2-y_2')/\ell^{2} \right\}.
\end{align}

\begin{figure}[t!]
\centering
\includegraphics[scale=1]{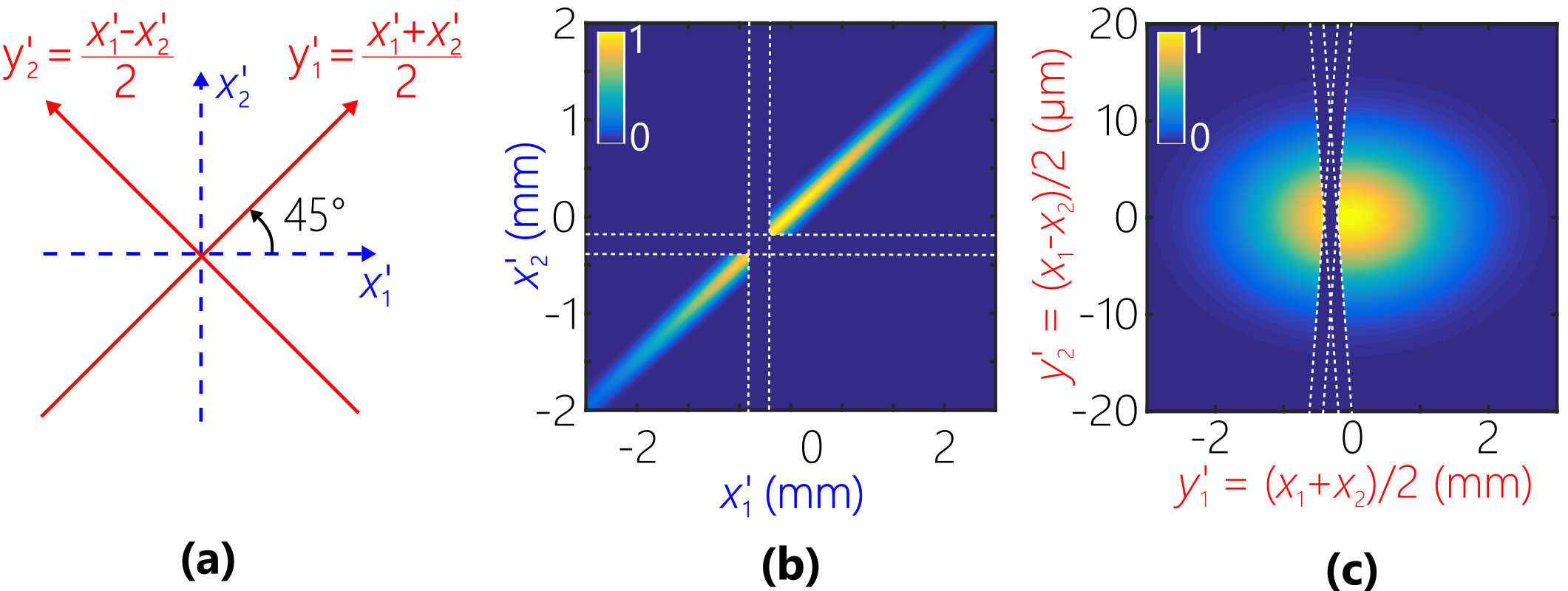}
\caption{(a) Illustration of rotated coordinates. An example of a generalized source is shown in (b) unrotated coordinates and (c) rotated coordinates.  For this example, the Gauss-Schell source parameters are $A\!=\!1$, $w\!=\!1$~mm, $\sigma\!=\!50~\mu$m.  The transmission function is such that
$t(x)\!\!=\!\!0$ for $x \!\in\! [a_1,a_2)$, and $t(x) \!\!=\!\! 1$ otherwise, where $a_1\!=\!-0.4$~mm and $a_2\!=\!-0.2$~mm. Dotted white lines indicate the regions affected by the transmission function.}
\label{Fig1:Overview}
\end{figure}

Underlying the construction of a generalized source in the next Section is a generic partially coherent field described by a Gauss-Schell model \cite{Carter:77}.  Our solution holds in the regime wherein the beam width of the source is assumed to be much larger than the coherence width, thus warranting the quasi-homogeneous approximation and yielding a coherence function $G(y_{1}',y_{2}')\!=\!I(y_1')g(y_2')$; where $I$ and $g$ are separable intensity and coherence functions, respectively \cite{Carter:77,Baleine:04}. Defining a Gaussian function $\Gaussian{\beta}{x}\!=\!\exp\{-x^{2}/2\beta^{2}\}$, then the coherence function of the Gauss-Schell model at $z\!=\!0$ is
\begin{align}\label{eqn:GaussSchell}
\CZeroUnmasked(y_1',y_2')=A_{\mathrm{o}}\, & \Gaussian{w_{\mathrm{o}}}{y_1'} \, \Gaussian{\sigma_{\mathrm{o}}}{y_2'},
\end{align}
where $A_{\mathrm{o}}$ is an amplitude, $w_{\mathrm{o}}$ the width of the intensity profile, and $\sigma_{\mathrm{o}}$ the coherence width of this initial field (all denoted with the subscript 'o'). A useful feature of the Gauss-Schell model is that its structure is propagation-invariant except for an overall phase. Indeed, after propagating a distance $d_{\mathrm{o}}$, Eq.~\ref{eqn:GaussSchell} takes the same form except for a phase factor,
\begin{align}\label{eqn:GaussSchellAfterPropagation}
\CZeroUnmasked(y_{1}',y_{2}')=A \, \exp\{iy_{1}'y_{2}'/R^{2}\} & \Gaussian{w}{y_{1}'} \, \Gaussian{\sigma}{y_{2}'},
\end{align}
where the modified Gauss-Schell parameters $A$, $w$, and $\sigma$, in addition to the new parameter $R$ (the radius of curvature of the quadratic phase), are related to the original parameters $A_{\mathrm{o}}$, $w_{\mathrm{o}}$, and $\sigma_{\mathrm{o}}$ through
\begin{subequations}
\begin{align}
A &=\frac{A_{\mathrm{o}}}{\sqrt{1+\xi_{\mathrm{o}}^{2}}}\,,\\
R &= \ell_{\mathrm{o}}\frac{\sqrt{1+\xi_{\mathrm{o}}^{2}}}{\xi_{\mathrm{o}}}\,, \\
w &=w_{\mathrm{o}}\sqrt{1+\xi_{\mathrm{o}}^{2}}\,\,, \\
\sigma &=\sigma_{\mathrm{o}}\sqrt{1+\xi_{\mathrm{o}}^{2}}\,\,,
\end{align}
\end{subequations}
where $\ell_{\mathrm{o}}\!=\!\sqrt{d_{\mathrm{o}}/2k}$, $\xi_{\mathrm{o}}$ is a unitless quantity given by $\xi_{\mathrm{o}}\!=\!\tfrac{\ell_{\mathrm{o}}^{2}}{w_{\mathrm{o}}\sigma_{\mathrm{o}}}\!=\!\frac{d_{\mathrm{o}}}{z_{\mathrm{GS}}}$, and $z_{\mathrm{GS}}\!=\!4\pi\frac{\sigma_{\mathrm{o}}w_{\mathrm{o}}}{\lambda}$ is an effective Rayleigh range for the Gauss-Schell model
(see \cite{FRIBERG1982383,GORI1983149} for an in-depth discussion of the free space propagation of a Gauss-Schell source).

We take the form in Eq.~\ref{eqn:GaussSchellAfterPropagation} to be the standard GS-model hereon, defined by four parameters $(A,R,w,\sigma)$. Any additional propagation of the GS-field does not change its form. Propagation a distance $d$ produces the same GS-model after transforming the parameters $(A,R,w,\sigma)\!\rightarrow\!(\tilde{A},\tilde{R},\tilde{w},\tilde{\sigma})$, with
\begin{subequations}\label{eqn:doublepropparameters}
\begin{align}
\tilde{A} &= \frac{A}{(1+\delta)\sqrt{1+\xi^{2}}}, \\
\tilde{R} &= R\,\sqrt{\frac{(1+\delta)(1+\xi^2)}{1+(1+\frac{1}{\delta})\xi^2}}, \\
\tilde{w} &= w\,(1+\delta)\sqrt{1+\xi^{2}},\\
\tilde{\sigma} &= \sigma\,(1+\delta)\sqrt{1+\xi^{2}},
\end{align}
\end{subequations}
where $\xi\!=\!\ell^{2}/\{w \sigma(1+\delta)\}\!=\!d/z_{\mathrm{GS}}$, $z_{\mathrm{GS}}\!=\!4\pi\sigma w(1+\delta)/\lambda$ is a scaled Rayleigh range, $\ell\!=\!\sqrt{d/2k}$, and $\delta\!=\!\ell^{2}/R^{2}$. In other words, after propagation in free space a distance $d$, the GS-field coherence becomes
\begin{align}\label{eqn:GaussSchellAfterPropagationDetector}
\widetilde{G}_d(y_1,y_2) = & \widetilde{A} \, \exp\big(iy_{1}y_{2}/\widetilde{R}^{2}\big) \Gaussian{\widetilde{w}}{y_{1}} \, \Gaussian{\widetilde{\sigma}}{y_{2}}.
\end{align}

\section{Propagation of Fields Produced by a Generalized Source}
\label{section:prop_complex_source}

We consider a partially coherent field modulated by a piecewise constant complex transmission function, referred to as a \emph{generalized source}.
The source $\CZeroUnmasked$ is masked by a piecewise constant transmission function $t(x) \in \mathbb{C}$, $|t(x)| \le 1$. Hence, the coherence function of the generalized source is
\begin{align}\label{eqn:complex_source}
G(y_1',y_2') = \CZeroUnmasked(y_1',y_2') \, t(y_1'+y_2') \, t^*(y_1'-y_2').
\end{align}
An example of such a source is shown in Fig. \ref{Fig1:Overview}(b) and (c) in both unrotated and rotated coordinates.

In this paper, we give a closed-form formula for the propagated coherence function of such sources that satisfy the following condition on the intensity and coherence widths
\begin{align}\label{eqn:AlphaSigmaRegime}
w > 10^2 \sigma > 10^3 \lambda .
\end{align}
This relation requires that intensity slowly varies with regard to the coherence width, and the coherence slowly varies with respect to the wavelength.
The transmission function $t$ is segmented into piecewise constant intervals $[a_j,a_{j+1})$,  $j=0,...,N$, where $-\infty=a_0 \!<\! a_1 \!<\! \cdots \!<\! a_{N} \!<\! a_{N+1}=\infty$. The theorem we prove below requires that the breakpoints and intervals satisfy the relations
\begin{subequations}\label{eqn:tRestrictions}
\begin{align}
\sum_{j=1}^{N}  \Gaussian{w}{|a_{j}|-3\sigma} &< 4, \label{eqn:tRestriction1} \\
\min_{j=2,N}(a_{j}-a_{j-1}) &> 3\sigma. \label{eqn:PartitionSize}
\end{align}
\end{subequations}
These relations put a limit on the resolution and number of features present in the transmittance function.
The first relation places a limit on the number of sections located close to the center of the field, while the second relation ensures that none of these sections is too small relative to the coherence width.
While sufficient but not necessary, the constraints in \eqref{eqn:tRestrictions} allow for a wide range of partially coherent sources of practical interest, as will be shown in the examples and numerical results below.

Before we state our main result, recall the definition of the Hilbert Transform of a square integrable function $f$,
\begin{align}\label{eqn:HilbertTransform}
Hf(\omega):=\mathrm{p.v.}{\frac{1}{\pi}}\int\frac{f(s)}{\omega-s}ds,
\end{align}
where the standard notation $\mathrm{p.v.}$ stands for principal value. 
For some real parameter $u$, we also define a conjugated Hilbert transform as
\begin{align}\label{eqn:ConjugatedHilbertTransform}
H^u f(\omega)&:=\exp(-i \omega u) \, H\left\{\exp(isu)\,f(s)\right\}(\omega)\nonumber\\
&=\exp(-i\omega u) \, \mathrm{p.v.}{\frac{1}{\pi}}\int\frac{\exp(isu)\,f(s)}{\omega-s}ds.
\end{align}

We proceed to our main result given in \eqref{eqn:mainresult}, which provides an effective approximation of the coherence function at a given distance from the generalized source. The formula \eqref{eqn:mainresult} characterizes the coherence function $G_d$ in terms of the coherence of the GS-field propagated a distance $d$ in free space and a multiplicative term -- expressed in terms of weighted conjugated Hilbert transforms of a Gaussian -- capturing the modification due to interaction with the transmission function. In obtaining our approximate formula, we consider the individual contributions of the different segments of the transmission function to the total coherence. This in turn yields an approximation to the coherence function based on Fourier transforms of truncated Gaussians giving rise to the conjugated Hilbert transform terms -- a relationship which has not been previously shown.

The technical bounds on the error of our approximate formula are provided in the proof in Appendix \ref{section:MainResultAppendix}. 
\begin{thmnonum}\label{main_result}
Let $\lambda$ be the wavelength, $w$ be the width of the beam intensity profile, and $\sigma$ be the transverse coherence width.
A generalized source as in \eqref{eqn:complex_source} satisfying \textnormal{Eqs.} (\ref{eqn:AlphaSigmaRegime},\ref{eqn:tRestrictions}) is situated at the plane $z\!=\!0$.  At the detection plane $z\!=\!d$, the coherence $G_d(y_1,y_2)$ is well approximated by
\begin{align}\label{eqn:mainresult}
\widetilde{G}_d(y_1,y_2) \frac{i}{2 \Gaussian{\eta \widetilde{\sigma}}{y_2}} \sum_{j=2}^{N} T_{j,j} \left[ \left(H^{b_{j}(y_1)} - H^{b_{j-1}(y_1)} \right)\GaussianNoParm{\widetilde{\sigma} / \eta}\right](y_2)
\end{align}
where $\widetilde{G}_d(y_1,y_2)$ is the coherence of the free 
propagating GS-field in \eqref{eqn:GaussSchellAfterPropagationDetector},
$T_{j,j} = |t(x)|^2$ for $x \in [a_{j-1},a_j)$, and
\begin{subequations}\label{eqn:mainresult_vars}
\begin{align}
\eta &= \sqrt{1+\frac{\sigma^2 \widetilde{\sigma}^2}{\ell^4}}, \\
b_j(y_1) &= \frac{1}{\eta^2 \ell^2} \left( a_j - \frac{y_1}{(1+\delta)(1+\xi^2)} \right).\label{bjs}
\end{align}
\end{subequations}
\end{thmnonum}
We note that \eqref{eqn:mainresult} recovers a close approximation to \eqref{eqn:GaussSchellAfterPropagationDetector} for the special case of uniform transmittance.  The difference is due to the finite extent of the source.
While \eqref{eqn:mainresult} shows that the contribution of each segment of the transmission function as in \eqref{eqn:AlphaSigmaRegime} and \eqref{eqn:tRestrictions} is essentially independent, note that we do not assume a priori independence in the contributions of the segments to the coherence.

As shown, information about the transmission function (the breakpoints $a_j$) is explicit in the parameters $b_j$ of the conjugated Hilbert Transform in \eqref{bjs} and the transmission coefficients $T_{j,j}$, wherefore the formula in \eqref{eqn:mainresult} is valuable in the inverse problem of recovering the transmission function from coherence measurements.

\section{Examples of Generalized Sources}\label{section:Examples}

\begin{figure}[t!]
\centering
\includegraphics[scale=1]{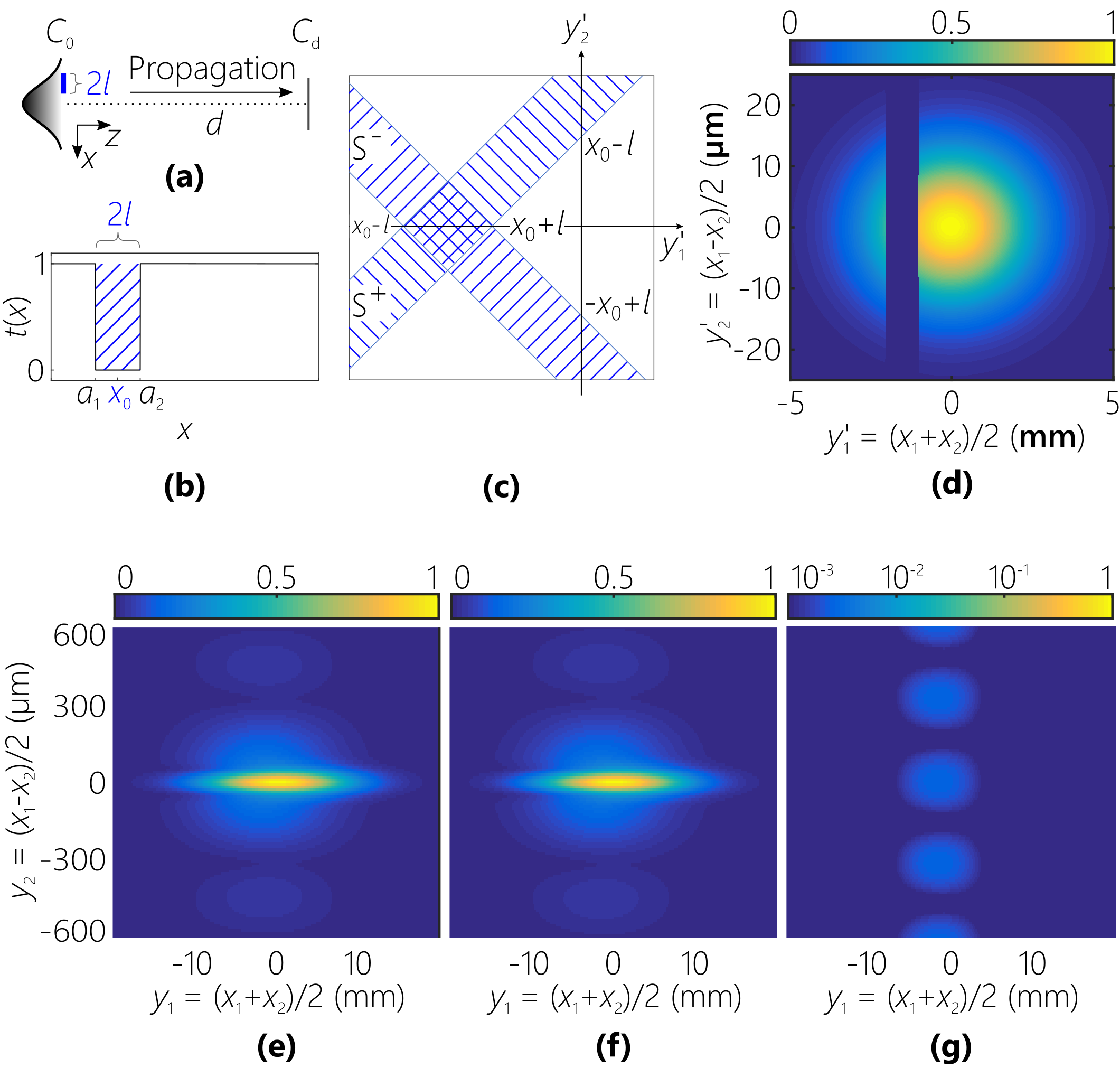}
  \caption{One object example with numerical results showing propagated coherence function in the plane at $z\!=\!1$~m.  (a) Diagram of the scenario.  (b) Transmission function. (c) Transmission function to be applied in coherence space.  The striped regions show the support of the inverted transmission function $1 \!-\! t(y_1' \!+\! y_2') t^*(y_1' \!-\! y_2')$.  (d) Modulus of source coherence function.  (e) Modulus of coherence function obtained using numerical integration of propagation function \eqref{fresneliny}. (f) Modulus of coherence function obtained using approximation \eqref{eqn:mainresult} of the theorem.  (g) Magnitude of error between complex coherences calculated by \eqref{fresneliny} and \eqref{eqn:mainresult}.  All plots are normalized against the maximum value attained in (e) and (f). (g) is plotted on a logarithmic scale to accentuate the small error.   It should be noted that the scale of the $y_2'$ axis is much smaller than the scale of the $y_1'$ axis, and so the ``strips'' mostly overlap in the plotted region.  The parameters for the source Gaussian are $A\!=\!1$, $w\!\approx\!1.7$~mm (yielding an intensity FWHM of 4 mm), $\sigma\!\approx\!8.5~\mu$m (yielding a coherence FWHM of 20 $\mu$m), and the source has no phase (i.e. in the limit as $R \rightarrow \infty$).  The wavelength is $\lambda=632$~nm.  The parameters for the object are $x_0\!=\!-1.5$~mm and $l\!=\!0.5$~mm. }
\label{Fig2:OneObject}
\end{figure}

For clarity of exposition, we analyze first the case of an object comprising a single segment ($N=2$) and then extend this to an example of an arbitrary generalized source.
In this scenario, we assume that a Gauss-Schell model field
exists at $z\!=\!0$.  The source is blocked by a single object centered along the transverse axis at $x\!=\!x_0$ with half-width $l$, and therefore its breakpoints are $a_1 \!=\! x_0\!-\!l$ and $a_2 \!=\! x_0\!+\!l$.  The scenario is depicted in Fig. \ref{Fig2:OneObject}(a) and the transmission function is shown in Fig. \ref{Fig2:OneObject}(b).
For this example, we consider the inversion of the transmission function in the coherence space $1 \!-\! t(y_1' \!+\! y_2') t^*(y_1' \!-\! y_2')$.  The inverted transfer function is chosen so that the coherence function $G$ is supported only on the union $S^+ \cup S^-$ of the ``strips'' shown in Fig. \ref{Fig2:OneObject}(c).
As described in Appendix \ref{section:MainResultAppendix}, these strips directly admit the closed-form solution presented in \eqref{eqn:mainresult}.  The source coherence function $G$ for this example is plotted in \ref{Fig2:OneObject}(d).
Fig. \ref{Fig2:OneObject}(e) shows the function obtained using numerical integration of \eqref{fresneliny}, and Fig. \ref{Fig2:OneObject}(f) shows the approximated results obtained using \eqref{eqn:mainresult}. 
The error in Fig. \ref{Fig2:OneObject}(g), which is plotted on a logarithmic scale, demonstrates good agreement between the exact and approximate equations.

\begin{figure}[t!]
\centering
\includegraphics[scale=1]{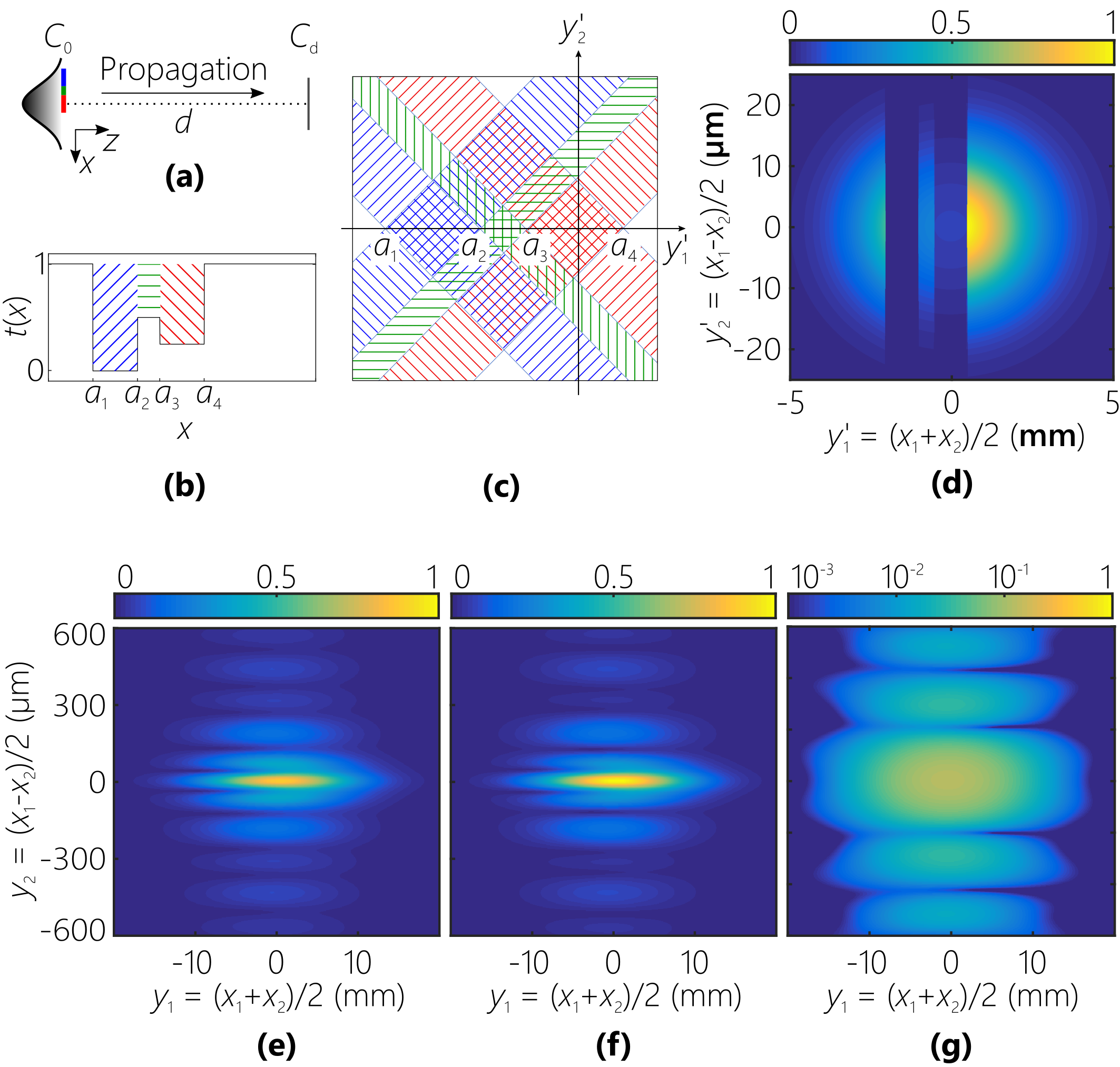}
  \caption{Generalized source example with numerical results showing propagated coherence function in the plane at $z\!=\!1$~m.  (a-g) are the same as in Fig. \ref{Fig2:OneObject}.  As with the one object example, the source parameters are $A\!=\!1$, $w\!\approx\!1.7$~mm, $\sigma\!\approx\!8.5~\mu$m, $\lambda=632$~nm, and no phase.  The breakpoints are at $a_1\!=\!-2$~mm, $a_2\!=\!-1$~mm, $a_3\!=\!-0.5$~mm, and $a_4\!=\!0.5$~mm with transmissions $t\left((-\infty,a_1)\right)\!=\!1$, $t\left([a_1,a_2)\right)\!=\!0$, $t\left([a_2,a_3)\right)\!=\!0.5$, $t\left([a_3,a_4)\right)\!=\!0.25$, and $t(\left[a_4,\infty)\right)\!=\!1$.}
\label{Fig3:GeneralizedSource}
\end{figure}

The next example demonstrates how the one object case naturally extends to more complicated transmission functions.
We will consider a similar scenario as for the previous example, except the transmission function has two additional sections (see Fig. \ref{Fig3:GeneralizedSource}(a) and (b)).
Each piecewise constant section $j$ of the transmission function influences two strip regions
\begin{align}
S_j^{ \pm }&=
\{ (y_1',y_2')\in\R^2,\; a_{j-1} \leq y_1' \mp y_2' \leq a_j \}.
\end{align}
As can be seen in Fig. \ref{Fig3:GeneralizedSource}(c), the interaction between these strips gives rise to $N^2$ piecewise constant sections in the coherence space transmittance. The theorem asserts that the only sections needed to form the approximation are those that fall on the $y_1$ axis.
Because the transmission function for this example is inverted, the true propagated output $G_d$ is given by
\begin{align}
G_d(y_1,y_2) = \CdUnmasked (y_1,y_2) - \overline{G}_d(y_1,y_2)
\end{align}
where $\CdUnmasked$ represents the propagated coherence of the unmasked Gaussian input function and $\overline{G}_d$ the propagated coherence function due to the inverted transmission function.
We show the numerical results in Fig. \ref{Fig3:GeneralizedSource}(e-g).

\begin{figure}[t!]
\centering
\includegraphics[scale=1]{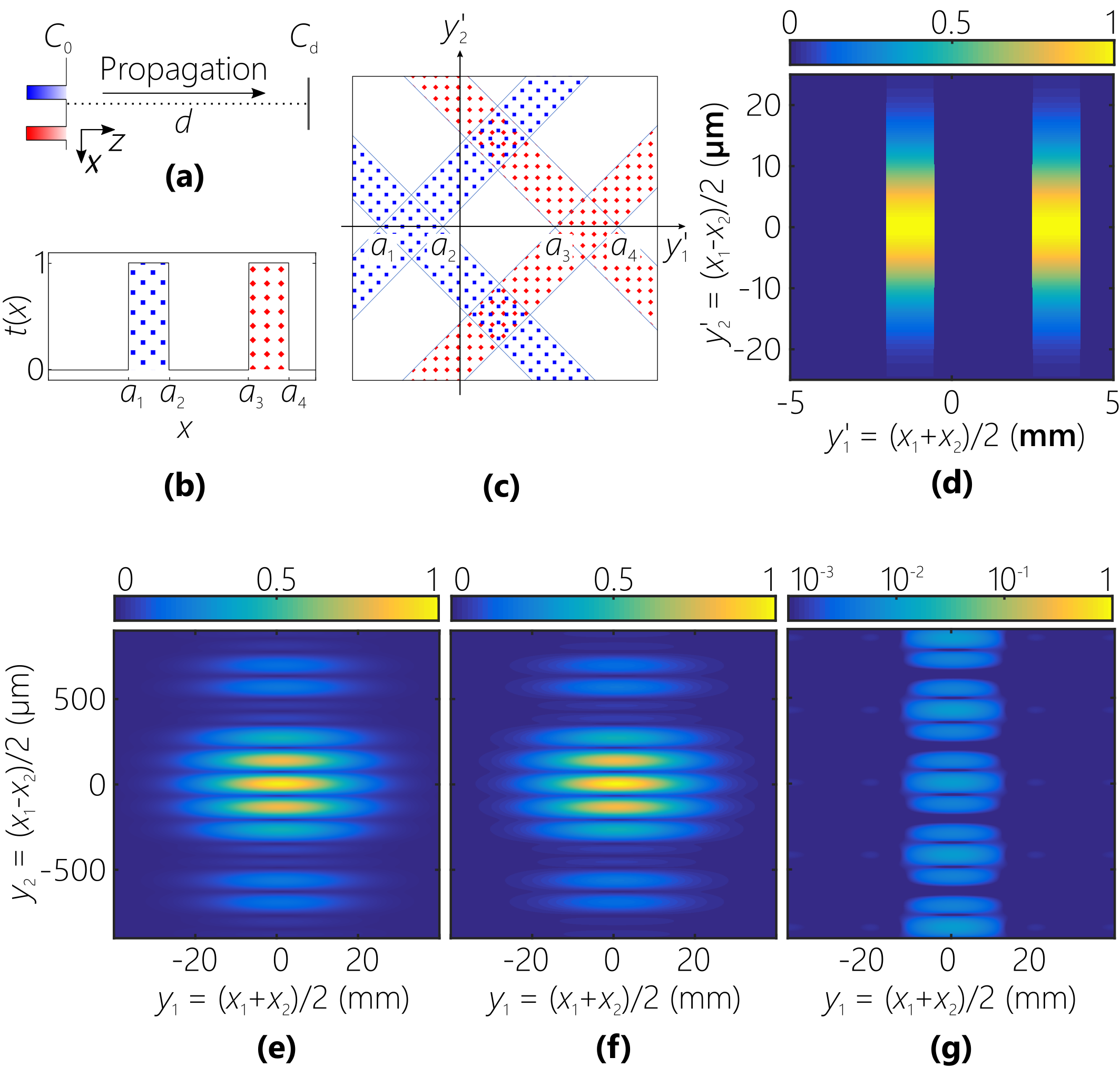}
  \caption{ Uniform source example with numerical results showing propagated coherence function in the plane at $z\!=\!2$~m.  (a-g) are the same as in Fig. \ref{Fig2:OneObject}, except that (c) shows strips due to the non-inverted transmission function $t(y_1' \!+\! y_2') t^*(y_1' \!-\! y_2')$.  The source parameters are $A=1$, $w=1$~m (thus approximating a uniform source), $\sigma\!\approx\!8.5~\mu$m, $\lambda=632$~nm, and no phase.  The breakpoints are at $a_1\!=\!-2$~mm, $a_2\!=\!-0.5$~mm, $a_3\!=\!2.5$~mm, and $a_4\!=\!4$~mm with transmissions $t\left((-\infty,a_1)\right)\!=\!t\left([a_2,a_3)\right)\!=\!t\left([a_4,\infty)\right)\!=\!0$, $t\left([a_1,a_2)\right)\!=\!t\left([a_3,a_4)\right)\!=\!1$.}
\label{Fig4:UniformSource}
\end{figure}

We present a final example demonstrating the approximation of a uniform source by a wide Gaussian (in this case $w\!=\!1$~m).  The source is shown in Fig. \ref{Fig4:UniformSource}.
Unlike the previous two examples, here in the coherence space we use the transmission function $t(y_1' \!+\! y_2') t^*(y_1' \!-\! y_2')$.
The numerical simulation is shown in Fig. \ref{Fig4:UniformSource}(e-g).
The numerical integration and approximated results are in very good agreement with a maximum error of $\approx 0.001$.

\section{Discussion}

We presented a closed-form approximation for the propagation of partially coherent fields emerging from a Gauss-Schell source modulated by a transmittance function.
Examples demonstrated the use of this framework in modeling beams propagating through apertures or past obstructing objects, with numerical results showing only small errors in the approximations.

Our work not only provides an efficient means of propagation in the forward model, but also lays a foundation for approaches to solve the inverse problem. Application  of these results to the inverse problem will be the subject of subsequent work.
For example, the closed-form solution derived in this paper can be used to guide an optimization algorithm in the reconstruction of a source based on measured coherence.
While back propagation of the detected coherence also allows for reconstruction of a scene, it typically requires measuring the entire coherence function \cite{ElHalawany17}. Therefore, our results hold promise to advance compressive methods for reconstructing the source profile using only few measurements, particularly when apriori information regarding the source is available. 
There are several possible directions for future research.  In the two-dimensional coherence model (i.e., one in which the field varies along both transverse axes), a four-fold integration is required to directly evaluate the Fresnel integral.  Our approach is expected to provide considerable computational advantages if extended to this domain.
Relaxing the restriction on the transmission function considered in \eqref{eqn:tRestrictions} would also be useful. Allowing an arbitrary number of short piecewise constant sections could approximate arbitrary smooth transmission profiles. Tightening the error bounds is another possible extension whereby we may be able to dispense with the verifiably conservative assumptions regarding the energy of the source.

\appendix

\section{The Fourier Transform of a Truncated Gaussian}\label{section:math_tools}
The proof of the theorem requires calculation of the Fourier transform of a truncated Gaussian.
In this section we derive the required results.

For $\sigma \!>\! 0$ define the Gaussian function as
$\Gaussian{\sigma}{x}\!=\!\exp\!\left\{-x^2 / 2\sigma^2 \right\}$.
For some $\sigma>0$ and $\omega\in\R$ let us consider the Fourier transform of a truncated Gaussian
\begin{align}\label{phi}
\Phi^\sigma (\omega,u)=\int_{-\infty}^u \exp(-i\omega x) \Gaussian{\sigma}{x} \, dx
\end{align}
and the \emph{cumulative distribution function}
\begin{align}\label{Phi0}
\Phi^\sigma_0(u):=\Phi^\sigma(0,u)=\int_{-\infty}^u \Gaussian{\sigma}{x} \, dx.
\end{align}

We first provide an exact formula for calculating \eqref{phi}.
Recall the Hilbert transform defined in \eqref{eqn:HilbertTransform} and the conjugated Hilbert transform defined in \eqref{eqn:ConjugatedHilbertTransform}.
The following result gives an exact formula for $\Phi$ in terms of the conjugated Hilbert transform.

\begin{lemma}\label{lemma_phi_hilbert_relation}
We have the following formula
\begin{align}\label{exactPhi} 
\Phi^\sigma(\omega,u)= \sqrt{\frac{\pi}{2}}\sigma\left[(I+iH^u)\GaussianNoParm{1 / \sigma}\right](\omega),
\end{align}
where $I$ stands for the identity operator and $H^u$ for the conjugated Hilbert transform in \eqref{eqn:ConjugatedHilbertTransform}. More explicitly,
\begin{align}
\Phi^\sigma(\omega,u) &= \sqrt{\frac{\pi}{2}} \sigma \left(\Gaussian{1 / \sigma}{\omega}
+ \, i \exp(-i\omega u)\frac{1}{\pi}p.v.\int\frac{\exp(isu)\Gaussian{1 / \sigma}{s}}{\omega-s}ds\right).
\end{align}
\end{lemma}

\begin{proof}
In the following, $\hat{f}$ denotes the Fourier transform of $f$.
\begin{align}
\Phi^\sigma (\omega,u)
& =\int \exp(-i\omega x)\IndicatorFunction_{(-\infty,u)}(x)\Gaussian{\sigma}{x}dx \nonumber\\
& =\widehat{\IndicatorFunction}_{(-\infty,u)}\star\widehat{N}^\sigma (\omega) \nonumber \\
& =\pi\left[\delta+ i \exp(-i(\cdot)u) H\right] \star \sqrt{2\pi}\sigma \Gaussian{1 / \sigma}{\omega} \nonumber \\
& =\sqrt{\frac{\pi}{2}} \sigma \left(\Gaussian{1 / \sigma}{\omega}
+ \, i \exp(-i\omega u){\frac{1}{\pi}}\mathrm{p.v.}\int\frac{\exp(isu)\Gaussian{1 / \sigma}{s}}{\omega-s}ds\right),
\end{align}where 
\begin{equation}
\IndicatorFunction_{(-\infty,u)}(x)=\left\{
\begin{array}{ll}
1,& x<u,\\
0&x>u.
\end{array}\right.
\end{equation}
In the third equality the Fourier transform is understood in the sense of (temperate)  distribution.
\end{proof}

We remark that from the properties of the Hilbert transform, it can be seen that the conjugated Hilbert transform obeys the inversion law
$
-(H^u)^2 f(x) = f(x).
$

Apart from the exact formula \eqref{exactPhi},
we are interested in an approximation with a form easier to handle analytically.
We now give an approximation formula together with the estimate in the error.

\begin{lemma}\label{lemma:Phiapprox}

\begin{align}\label{Phiexpanded}
&\Phi^\sigma(\om,u)
= \Gaussian{1/\sigma}{\om} \Phi^\sigma_0(u) \left[1 + (\i \sigma^2 \omega) \exp(-iu\om)\frac{\Gaussian{\sigma}{u}}{\Phi^\sigma_0(u)}R(\omega,u,\sigma)\right],
\end{align}
where 
\begin{align}
R(\omega,u,\sigma)=\int_0^1\exp(-iu\omega x)
\exp\!\left( \frac{\sigma^2\omega^2}{2}x^2 \right) dx
\end{align}
In particular, we have
\begin{align}\label{Phiapprox}
&\Phi^\sigma(\om,u) \approx \Gaussian{1/\sigma}{\om} \Phi^\sigma_0(u)
\end{align}
provided
\begin{align}\label{PhiapproxCondition}
\left| \sigma^2 \omega \frac{\Gaussian{\sigma}{u}}{\Phi_0^\sigma(u)} \right| \ll 1.
\end{align}
\end{lemma}

\begin{proof}

We will show that $\omega\mapsto \Phi^\sigma(\omega,u)$  satisfies the linear differential equation
\begin{align}
\frac{d\Phi^\sigma(\omega,u)}{d\om}+\sigma^2\om\Phi^\sigma(\omega,u)= \i \sigma^2 \exp(-\i u \om) \Gaussian{\sigma}{u}.
\end{align}
Indeed, by differentiation,
\begin{align}
\frac{d\Phi^\sigma}{d\om}
&= \i \sigma^2 \int_{-\infty}^u \exp(-\i x \om) \, {\frac{d}{dx}}\!\left[ \exp\!\left(-\frac{x^2}{2 \sigma^2}\right) \right] \d x \nonumber \\
&= \i \sigma^2 \exp\!\left(-\i x \om\right) \exp\!\left( -\frac{x^2}{2 \sigma^2} \right)
\bigg\rvert^{x=u}_{x \rightarrow -\infty} \nonumber \\
&\qquad - \i \sigma^2 (- \i \om) \int_{-\infty}^u \exp(-\i x \om) \exp\!\left(-\frac{x^2}{2 \sigma^2}\right) \d x \nonumber \\
&= \i \sigma^2 \exp(-\i u \om) \Gaussian{\sigma}{u} - \sigma^2 \om \Phi^\sigma (\om, u) \d x.
\end{align}
Using the integrating factor $\exp\!\left(\sigma^2 \om^2 / 2\right)$, an integration from $0$ to $\sigma$, and a scaling by a factor of $\sigma$ in the ensuing integral, the formula \eqref{Phiexpanded} is obtained.
We estimate $R$ as follows
\begin{align}
\left| R(\om,\sigma,u)\right|
&\leq \int_{0}^{1}\exp\!\left(\frac{\sigma^2 \om^2}{2}x^2\right) \d x \nonumber \\
&\leq \int_{0}^{1}\exp\!\left(\frac{\sigma^2 \om^2}{2}x\right) \d x \nonumber \\
&=\frac{2}{\sigma^2 \omega^2}\left[\exp\!\left(\frac{\sigma^2 \om^2}{2}\right)-1\right],
\end{align}
from which the estimate \eqref{Phiapprox} and condition \eqref{PhiapproxCondition} follow.

\end{proof}

\section{Proof of Theorem}\label{section:MainResultAppendix}

\begin{figure}[t!]
\centering
\includegraphics[scale=1]{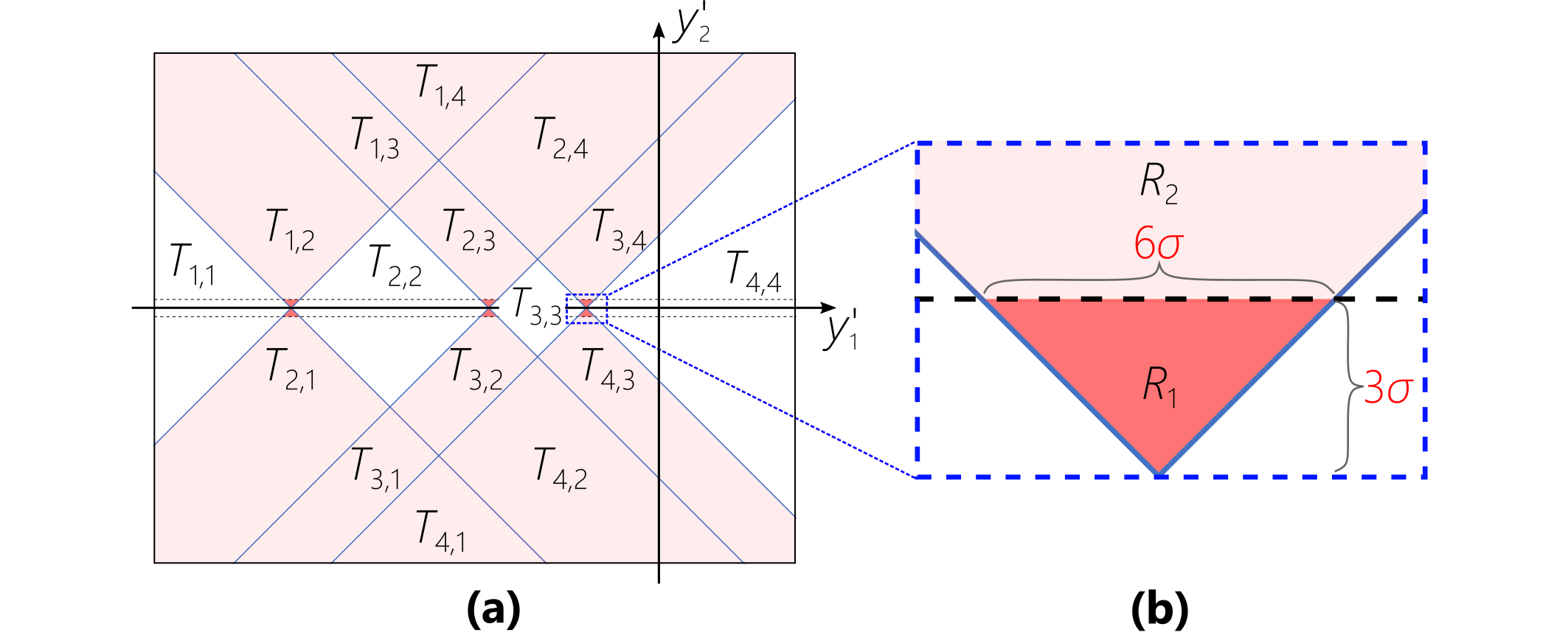}
  \caption{Regions of approximation for Lemma \ref{ApproxStrips}.  (a) shows the coefficients associated with each piecewise constant section.  The region of approximation $R_1$ is shaded dark red, while the region $R_2$ is shaded light red.  (b) provides a detailed view of one of the triangle regions making up region $R_1$. }
\label{Fig5:ApproxRegions}
\end{figure}

Next, we provide a proof for the theorem stated in Section \ref{section:prop_complex_source}.
The transmission function in the coherence space is
\begin{align}\label{eqn:BetaCoverOfT}
t(x_1') t^*(x_2') &= \sum_{j,k=1}^{N} T_{j,k} \,\, \IndicatorFunction_{B_{j,k}}\!(x_1',x_2')
\end{align}
where $\IndicatorFunction$ denotes the indicator function and $T_{j,k}$ denotes the transmissivity coefficient within region
\begin{align}
B_{j,k} = [a_{j-1},a_j) \times [a_{k-1},a_k), \quad 1 \le j,k \le  N+1.
\end{align}
The coefficients are shown in Fig. \ref{Fig5:ApproxRegions}.
As with the examples in Section \ref{section:Examples}, without loss of generality we may also use the transmission function $1 \!-\! t(x_1') t^*(x_2')$.
In terms of the unmasked Gaussian beam $\CZeroUnmasked(y_1',y_2')=A \, \exp\{iy_{1}'y_{2}'/R^{2}\} \, \Gaussian{w}{y_1'} \, \Gaussian{\sigma}{y_2'}$, the propagated coherence function is
\begin{align}\label{CoherenceOutput}
G_d(y_1,y_2)
&= \frac{1}{2 \pi \ell^2} \sum_{j,k=1}^{N+1} T_{j,k} \iint_{B_{j,k}} \CZeroUnmasked(y_1', y_2')\,  \mathcal{L}(y_1,y_1',y_2,y_2') \,\, dy_1' \,\, dy_2'.
\end{align}
We first apply Lemma \ref{ApproxStrips}, which allows the source coherence to be approximated by a series of infinite strips.
\begin{lemma}\label{ApproxStrips}Assume the hypotheses in the theorem.
The propagated coherence function \eqref{CoherenceOutput} can be approximated by
\begin{align}\label{CoherenceOutputApproxStrips}
G_d(y_1,y_2)
&\approx \frac{1}{4 \pi \ell^2} \sum_{j = 2}^{N} T_{j,j}
\left( \iint_{S_j^{+}} \!\CZeroUnmasked(y_1', y_2') \, \mathcal{L}(y_1,y_1',y_2,y_2') \, \, dy_1' \, dy_2' \right. \nonumber \\
&\qquad\qquad\qquad\quad\quad + \left. \iint_{S_j^{-}} \!\CZeroUnmasked(y_1', y_2') \, \mathcal{L}(y_1,y_1',y_2,y_2') \, \, dy_1' \, dy_2' \right)
\end{align}
where we use the two strip regions $S_{j}^+ \!\!=\!\! \bigcup_{k = 1}^{N+1} B_{j,k}$ and $S_{j}^- \!\!=\!\! \bigcup_{k = 1}^{N+1} B_{k,j}$.
Moreover, the magnitude of the pointwise error in this approximation is bounded from above by $\left[ 40\sigma \!+\! 4 \sqrt{2\pi} w \Phi_0^1(-3) \right] A\sigma$.
\end{lemma}

\begin{proof}
We will assume for our source function $G$, that $T_{1,1} \!=\! T_{1,N\!+\!1} \!=\! T_{N\!+\!1,1} \!=\! T_{N\!+\!1,N\!+\!1} \!=\! 0$.  Since the source is Gaussian, an appropriate truncation (at $|y_1'| = 3w$ for example) will result in only a small error.

We start by approximating the source coherence by
\begin{align}\label{eqn:SourceApprox}
G'(y_1',y_2')
= \sum_{j,k=2}^{N} \frac{T_{j,j}+T_{k,k}}{2} \, \IndicatorFunction_{B_{j,k}}\!(y_1'+y_2',y_1'-y_2') \, \CZeroUnmasked(y_1',y_2')
\end{align}
Then we have
\begin{align}\label{eqn:SourceApproxTerms}
G(y_1',y_2') = G'(y_1',y_2')
+ R(y_1',y_2')
\end{align}
where
\begin{align}
R(y_1',y_2')
&=\sum_{\substack{j,k = 2 \\ j \ne k}}^{N} \left( \frac{T_{j,k} + T_{k,j}}{2} - T_{j,j} \right)
\CZeroUnmasked(y_1', y_2') \nonumber \\
& = R_1(y_1',y_2') + R_2(y_1',y_2').
\end{align}
Error terms $R_1$ and $R_2$ arise from different regions of the source as illustrated in Fig. \ref{Fig5:ApproxRegions}.

The term $R_1$ comes from integration over the region $\bigcup_{j,k = 2}^{N}\left( B_{j,k}\bigcap \{ |y_2'| \le 3\sigma \}\right) $ of small triangles as in Fig. \ref{Fig5:ApproxRegions}(b).
The coherence function at $z\!=\!d$ due to this term can be bounded by
\begin{align}
|R_{1,d}(y_1,y_2)| &= \left| \iint R_1(y_1',y_2') \mathcal{L}(y_1,y_1',y_2,y_2') \, dy_1' \, dy_2' \right| \nonumber \\
&\le \iint |R_1(y_1',y_2')| \, dy_1' \, dy_2' \nonumber \\
&\le A \sum_{j=1}^N \left| T_{j,j+1} + T_{j+1,j} - T_{j,j} - T_{j+1,j+1} \right|  \nonumber \\
&\qquad \times \sum_{k=0}^{11} \int_{\frac{k}{4}\sigma}^{\frac{k+1}{4}\sigma} \int_{a_j-\frac{k+1}{4}\sigma}^{a_j+\frac{(k+1)}{4}\sigma} \Gaussian{w}{y_1'} \Gaussian{\sigma}{y_2'} \,\, dy_1' \,\, dy_2' \nonumber \\
&\le \frac{5}{2} A\sigma^2 \sum_{j=1}^N \left| T_{j,j+1} \!+\! T_{j+1,j} \!-\! T_{j,j} \!-\! T_{j+1,j+1} \right| \nonumber \\
&\qquad \times \Gaussian{w}{|a_j|\!-\!3\sigma} \nonumber \\
&\leq 40 A \sigma^2,\label{eqn:R1}
\end{align}
where the last inequality uses the hypothesis \eqref{eqn:tRestriction1}.

The term $R_2$ comes from the region $\bigcup_{j,k = 2}^{N}\left[ \left( B_{j,k} \cup B_{k,j} \right) \bigcap \{ |y_2'| > 3\sigma \} \right]$.  Making use of the inequality $\left| \frac{T_{j,k} + T_{k,j}}{2} - T_{j,j} \right| \le 2$ for any $j,k$, we have
\begin{align}
|R_{2,d}(y_1,y_2)| &= \left| \iint R_2(y_1',y_2') \mathcal{L}(y_1,y_1',y_2,y_2') \, dy_1' \, dy_2' \right| \nonumber \\
&\le \iint |R_2(y_1',y_2')| \, dy_1' \, dy_2' \nonumber \\
&\le 4 A \int_{3 \sigma}^{\infty} \int_{-\infty}^{\infty} \Gaussian{w}{y_1'} \Gaussian{\sigma}{y_2'} \, dy_1' \, dy_2' \nonumber \\
&= 4 \sqrt{2\pi} A w \int_{3\sigma}^{\infty} \Gaussian{\sigma}{y_2'} \, dy_2' \nonumber \\
&= 4 \sqrt{2\pi} A w \sigma \, \Phi_0^1(-3)
\end{align}
If $w$ is large and the transmission function $t$ is zero outside the interval $[a_1, a_N)$, then we may instead bound the error $R_2$ by
\begin{align}
|R_{2,d}(y_1,y_2)|
&\le 4 A (a_N-a_1) \int_{3 \sigma}^{\infty} \Gaussian{\sigma}{y_2'} \, dy_2' \nonumber \\
&= 4 \sqrt{2\pi} A (a_N-a_1) \sigma \, \Phi_0^1(-3).
\end{align}

\end{proof}

Finally, the following lemma can be applied to reduce the propagation integrals over the strips to a closed-form.  Since the Fresnel approximation is assumed, in the following proof we use the fact that $y_1,y_2 \!\ll\! d$, and that the source must be concentrated about the origin in the $y_1',y_2'$-plane.

\begin{lemma}\label{mainlemma}
The contribution to the detected coherence due to the $j$-th strips can be approximated as
\begin{align}
G_{d,j}^{\pm}(y_1,y_2) &= \frac{1}{2 \pi \ell^2} \iint_{S_j^{\pm}} \CZeroUnmasked (y_1', y_2') \, \mathcal{L}(y_1,y_1',y_2,y_2') \, dy_1' \, dy_2' \nonumber \\
&\approx \widetilde{G}_d(y_1,y_2) \, \frac{i}{2 \Gaussian{\eta \widetilde{\sigma}}{y_2}} \nonumber\\
&\qquad\times\left[ \pm \left( H^{\pm b_{j}(y_1)} - H^{\pm b_{j-1}(y_1)} \right)N^{\widetilde{\sigma} / \eta} \right](y_2),\label{D1}
\end{align}
where $\widetilde{G}_d$ is defined in \eqref{eqn:GaussSchellAfterPropagationDetector}, and the variables $\widetilde{\sigma}$, $\eta$, and $b_j$ are as defined in Eqs. (\ref{eqn:doublepropparameters},\ref{eqn:mainresult_vars}).
\end{lemma}

\begin{proof}
We perform the integration over the strip
\begin{align}
S_j^{+} =
\{ (y_1',y_2')\in\R^2,\; a_{j-1} \leq y_1' - y_2' < a_j \}.
\end{align}
The calculation over the  strip
\begin{align}
S_j^{-} =
\{ (y_1',y_2')\in\R^2,\; a_{j-1} \leq y_1' + y_2' < a_j \}
\end{align}
follows similarly and is not detailed.

\begin{align}
G_{d,j}^{+}(y_1,y_2)
&= \frac{1}{2 \pi \ell^2} \iint_{S_j^{+}} \CZeroUnmasked (y_1',y_2') \, \mathcal{L}(y_1,y_1',y_2,y_2') \, \d y_1' \, \d y_2' \nonumber \\
&= \frac{1}{2 \pi \ell^2} A
\exp(\i y_1 y_2 / \ell^2) \int \exp(-\i y_1' y_2 / \ell^2) \Gaussian{w}{y_1'} \nonumber \\
&\qquad\times \int_{y_1'-a_j}^{y_1'-a_{j-1}} \!\!\!\!\exp\!\left(- \i y_2' \left[ y_1 \!-\! (1+\delta) y_1' \right] / \ell^2 \right) \Gaussian{\sig}{y_2'} \, \d y_2' \, \d y_1' \nonumber \\
&= \frac{1}{2 \pi \ell^2} A \exp(\i y_1 y_2 / \ell^2)\int  \exp(-\i y_1' y_2 / \ell^2) \Gaussian{w}{y_1'} \nonumber \\
 &\qquad\times\left[  \Phi^{\sig} \! \left(\frac{y_1 - \left(1+\delta \right) y_1'}{\ell^2}, y_1'-a_{j-1} \right) \right. \nonumber \\
 & \left. \qquad\qquad\qquad - \, \Phi^{\sig} \! \left( \frac{y_1 - \left( 1+\delta \right) y_1'}{\ell^2}, y_1'-a_j \right)  \right]\d y_1' \nonumber \\ 
&\approx \frac{1}{2 \pi \ell^2}A \exp(\i y_1 y_2 / \ell^2)\int \exp(-\i y_1' y_2 / \ell^2) \nonumber \\
&\qquad\qquad\qquad\times \Gaussian{w}{y_1'} \Gaussian{\ell^2 / \sig(1+\delta)}{y_1' - \frac{y_1}{1+\delta}} \nonumber \\
&\qquad\qquad\qquad\times\left[  \Phi_0^{\sig}(y_1'-a_{j-1})
 - \Phi_0^{\sig} (y_1'-a_j)  \right]\d y_1' \nonumber \\
&=\frac{1}{2 \pi \ell^2} A \exp(\i y_1 y_2 / \ell^2) \Gaussian{\widetilde{w}}{y_1} \nonumber \\
&\qquad\qquad\qquad\times \int \exp(-\i y_1' y_2 / \ell^2)  \Gaussian{\ell^2 / \widetilde{\sigma}}{y_1' - c_1} \nonumber \\
&\qquad\qquad\qquad\times \left[ \Phi_0^{\sig}(y_1'-a_{j-1})
 - \Phi_0^{\sig} (y_1'-a_j)  \right]\d y_1',
\end{align}
where
\begin{subequations}
\begin{align}
{c_1} &= \frac{y_1}{(1+\delta)(1 + \xi^2)}.
\end{align}
\end{subequations}
Due to the nature of the inner integral of the second equality, the outer integral is effectively truncated such that $-3\sigma \!<\! y_1'-a_{j-1} \!<\! y_1'-a_j \!<\! 3\sigma$.  Therefore, with the hypotheses \eqref{eqn:AlphaSigmaRegime} and Eqs. (\ref{eqn:tRestrictions}), the approximation formula \eqref{Phiapprox} applies since
\begin{align}
\frac{\sigma^2}{\ell^2} \left[ y_1 - (1+\delta) y_1' \right] \frac{\Gaussian{\sigma}{y_1'-a_k}}{\phi_0^\sigma(y_1'-a_k)}  \ll 1 \:
\end{align}
for $k\!=\!j\!-\!1,j$.  Substituting $y_1' = y_1'' + c_1$, we continue
\begin{align}
G_{d,j}^{+}(y_1,y_2)
&\approx\frac{1}{2 \pi \ell^2} A \exp(\i y_1 y_2 / \ell^2) \Gaussian{\widetilde{w}}{y_1} \nonumber \\
&\qquad\times \int \exp\!\left( -\i y_2 (y_1'' + c_1)/\ell^2\right) \Gaussian{\ell^2 / \widetilde{\sigma}}{y_1''} \nonumber \\
&\qquad\times \left[  \Phi_0^{\sig}(y_1''+(c_1-a_{j-1})) \right. \nonumber \\
&\left. \qquad\qquad\qquad - \Phi_0^{\sig} (y_1''+(c_1-a_j))  \right]\d y_1'' \nonumber\\
&=\widetilde{G}_d(y_1,y_2) \frac{\widetilde{\sigma} }{\sqrt{2\pi} \eta \, \Gaussian{\eta \widetilde{\sigma}}{y_2}} \nonumber \\
&\qquad\times
\left[ \Phi^{\eta/\widetilde{\sigma}} \! \left(y_2,\frac{a_{j}-c_1}{\eta^2 \ell^2}\right) - \Phi^{\eta/\widetilde{\sigma}} \! \left(y_2,\frac{a_{j-1}-c_1}{\eta^2 \ell^2}\right)  \right]
\nonumber \\
&= \widetilde{G}_d(y_1,y_2) \frac{i}{2 \Gaussian{\eta \widetilde{\sigma}}{y_2}} \nonumber \\
&\qquad\times \left[ \left(H^{b_{j}(y_1)} - H^{b_{j-1}(y_1)} \right)\GaussianNoParm{\widetilde{\sigma} / \eta}\right](y_2).
\end{align}

\end{proof}

\section*{Funding.}
Defense Advanced Research Projects Agency (DARPA), Defense Science Office under contract HR0011-16-C-0029.

\bibliographystyle{IEEEtran}
\bibliography{main}

\end{document}